\documentclass[12pt]{article}%

\ifx\SHORTVER\undefined
\newcommand{\InFullVer}[1]{#1}
\newcommand{\InShortVer}[1]{}
\else
\newcommand{\InFullVer}[1]{}
\newcommand{\InShortVer}[1]{#1}
\fi

\ifx\LongSubmitVer\undefined
\newcommand{\InLongSubmitVer}[1]{}
\else
\newcommand{\InLongSubmitVer}[1]{#1}
\fi

\usepackage[cm]{fullpage}
\usepackage{microtype}
\usepackage{hyperref}%
\usepackage{mleftright}%
\hypersetup{%
   breaklinks,%
   ocgcolorlinks,
   colorlinks=true,%
   linkcolor=[rgb]{0.45,0.0,0.0},%
   citecolor=[rgb]{0,0,0.45}
}
\usepackage{graphicx}
\usepackage{xspace}
\usepackage{color}
\usepackage{amsmath}
\usepackage{paralist}
\usepackage{euscript}

\usepackage{stmaryrd}
\usepackage{amssymb}

\usepackage{ntheorem}
\theoremseparator{.}

\usepackage[compact]{titlesec}
\titlelabel{\thetitle. }

\InLongSubmitVer{%
   \usepackage{fancyhdr}%
}

\newcommand{\brc}[1]{\left\{ {#1} \right\}}

\newcommand{\MakeBig}{\rule[-.2cm]{0cm}{0.4cm}}

\newcommand{\si}[1]{#1}
\renewcommand{\th}{th\xspace}

\newcommand{\atgen}{\symbol{'100}}

\newcommand{\SarielThanks}[1]{%
   \thanks{%
      Department of Computer Science; %
      University of Illinois; %
      201 N. Goodwin Avenue; %
      Urbana, IL, %
      61801, USA; %
      {\tt \si{sariel}\atgen{}\si{illinois.edu}}; %
      {\tt \url{http://sarielhp.org}.} #1 } }

\definecolor{blue25}{rgb}{0.0,0,0.85}
\newcommand{\emphic}[2]{%
     \textcolor{blue25}{%
         \textbf{\emph{#1}}}%
         \index{#2}}

\newcommand{\ceil}[1]{\left\lceil {#1} \right\rceil}

\newcommand{\emphi}[1]{\emphic{#1}{#1}}

\newcommand{\cardin}[1]{\left| {#1} \right|}

\newtheorem{theorem}{Theorem}[section]
\newtheorem{lemma}[theorem]{Lemma}%

\theoremstyle{remark}%
\theoremheaderfont{\sf}
\theorembodyfont{\upshape}%

\newtheorem{defn}[theorem]{Definition}

\newtheorem*{remark:unnumbered}[theorem]{Remark}%
\theoremheaderfont{\em}%
\theorembodyfont{\upshape}%
\theoremstyle{nonumberplain}
\theoremseparator{}
\theoremsymbol{\myqedsymbol}
\newtheorem{proof}{Proof:}

\newcommand{\HLinkShort}[2]{\hyperref[#2]{#1\ref*{#2}}}
\newcommand{\HLink}[2]{\hyperref[#2]{#1~\ref*{#2}}}
\newcommand{\HLinkPage}[2]{\hyperref[#2]{#1~\ref*{#2}%
      $_\text{p\pageref{#2}}$}}
\newcommand{\HLinkPageOnly}[1]{\hyperref[#1]{Page~\refpage*{#1}%
      $_\text{p\pageref{#1}}$}}

\newcommand{\HLinkSuffix}[3]{\hyperref[#2]{#1\ref*{#2}{#3}}}
\newcommand{\HLinkPageSuffix}[3]{\hyperref[#2]{#1\ref*{#2}%
      #3$_\text{p\pageref{#2}}$}}

\newcommand{\figlab}[1]{\label{fig:#1}}
\newcommand{\figref}[1]{\HLink{Figure}{fig:#1}}

\newcommand{\figrefpage}[1]{\HLinkPage{Figure}{fig:#1}}

\newcommand{\seclab}[1]{\label{sec:#1}}
\newcommand{\secref}[1]{\HLink{Section}{sec:#1}}

\newcommand{\lemlab}[1]{\label{lemma:#1}}
\newcommand{\lemref}[1]{\HLink{Lemma}{lemma:#1}}%

\newcommand{\thmlab}[1]{{\label{theo:#1}}}
\newcommand{\thmref}[1]{\HLink{Theorem}{theo:#1}}

\providecommand{\eqlab}[1]{}%
\renewcommand{\eqlab}[1]{\label{equation:#1}}

\newcommand{\Eqrefpage}[1]{\HLinkPageSuffix{Eq.~(}{equation:#1}{)}}

\newcommand{\permut}[1]{\left\langle {#1} \right\rangle}%
\newcommand{\pth}[1]{\mleft({#1}\mright)}
\newcommand{\pbrc}[2][\!\!]{#1\left[ {#2} \MakeBig \right]}

\newcommand{\Tree}{\mathcal{T}}%

\newcommand{\metricC}{d}

\newcommand{\dist}[2]{\metricC\pth{#1, #2}}

\newcommand{\Graph}{{G}}

\newcommand{\Vertices}{V}
\newcommand{\Edges}{E}
\newcommand{\EdgesX}[1]{E\pth{#1}}

\newcommand{\eps}{{\varepsilon}}%

\newcommand{\net}{\mathcal{N}}
\newcommand{\Snet}{\mathcal{S}}

\newcommand{\DiameterX}[1]{\mathrm{diam}\pth{#1}}
\newcommand{\Diameter}{\Delta}
\newcommand{\Spread}{\Phi}

\newcommand{\Mtr}{\mathcal{M}}
\newcommand{\Event}{\EuScript{E}}

\newcommand{\distG}[2]{{d}_\Graph \pth{ #1, #2}}

\newcommand{\distC}[1]{\mathrm{\ell}_{#1}}

\newcommand{\etal}{\textit{et~al.}\xspace}

\newcommand{\Boruvka}{Bor\r{u}vka\xspace}

\newcommand{\SpreadX}[1]{\mathrm{spread}\pth{#1}}

\newcommand{\Otilde}{\widetilde{O}}

\newcommand{\ropt}{r_{\mathrm{opt}}}

\newcommand{\Set}[2]{\left\{ #1 \;\middle\vert\; #2 \right\}}

\newcommand{\DGraph}[1]{\mathcal{D}_{#1}}
\newcommand{\bdX}[1]{\partial{#1}}
\newcommand{\remove}[1]{}

\newcommand{\PntSet}{P}

\renewcommand{\Re}{\mathbb{R}}

\newcommand{\IntRange}[1]{\left\llbracket #1 \right\rrbracket}
\newcommand{\ShortPairs}[1]{P_{\leq {#1}}}

\newcommand{\SetPageNumberStyle}{}
\InLongSubmitVer{%
   \renewcommand{\SetPageNumberStyle}{%
      \pagestyle{fancy}%
      \thispagestyle{fancy}%
      \renewcommand{\headrulewidth}{0pt}
      \fancyhf{}%
      \cfoot{Page \thepage{} of the full version}%
   }%
}

\SetPageNumberStyle{}

\begin{document}

\title{Approximate Greedy Clustering and Distance Selection for Graph
   Metrics}

\author{%
   David Eppstein%
   \thanks{%
      Computer Science Dept., Univ. of California, Irvine;
      \si{eppstein}\atgen{}\si{uci.edu};
      {\tt\url{http://www.ics.uci.edu/\string~eppstein}}.  Work
      supported in part by the National Science Foundation under
      grants 0830403 and 1217322, and by the Office of Naval Research
      under \si{MURI} grant N00014-08-1-1015.}~~~~~~~~~~%
   \and%
   Sariel Har-Peled%
   \SarielThanks{%
      Work on this paper was partially supported by a NSF AF awards
      CCF-0915984 and CCF-1217462.}%
   \and%
   Anastasios Sidiropoulos\thanks{Dept.\ of Computer Science and
      Engineering and Dept.\ of Mathematics, The Ohio State
      University, Columbus, OH 43210; {\tt sidiropoulos.1@osu.edu};
      {\tt\url{http://sidiropoulos.org}}. Supported in part by David
      and Lucille Packard Fellowship, NSF AF award CCF-0915984, and
      NSF grants CCF-0915519 and CCF-1423230.}  }

\date{\today}

\SetPageNumberStyle{}
\maketitle
\SetPageNumberStyle{}

\InShortVer{%
   \setcounter{page}{0} 
   \thispagestyle{empty}%
}

\begin{abstract}
    In this paper, we consider two important problems defined on
    finite metric spaces, and provide efficient new algorithms and
    approximation schemes for these problems on inputs given as graph
    shortest path metrics or high-dimensional Euclidean metrics.  The
    first of these problems is the greedy permutation (or
    farthest-first traversal) of a finite metric space: a permutation
    of the points of the space in which each point is as far as
    possible from all previous points. We describe randomized
    algorithms to find $(1+\eps)$-approximate greedy permutations of
    any graph with $n$ vertices and $m$ edges in expected time
    $O\pth{\eps^{-1}(m+n)\log n\log(n/\eps)}$, and to find
    $(1+\eps)$-approximate greedy permutations of points in
    high-dimensional Euclidean spaces in expected time
    $O(\eps^{-2} n^{1+1/(1+\eps)^2 + o(1)})$.  Additionally we
    describe a deterministic algorithm to find exact greedy
    permutations of any graph with $n$ vertices and treewidth $O(1)$
    in worst-case time $O(n^{3/2}\log^{O(1)} n)$.  The second of the
    two problems we consider is distance selection: given
    $k \in \IntRange{ \binom{n}{2} }$, we are interested in computing
    the $k$\th smallest distance in the given metric space.  We show
    that for planar graph metrics one can approximate this distance,
    up to a constant factor, in near linear time.
\end{abstract}

\InShortVer{\newpage}


\section{Introduction}

In this paper we are interested in several important algorithmic
problems on finite metric spaces, including the construction of greedy
permutations, the problem of selecting the $k$\th distance among all
pairs of points in the space, and the problem of counting the number
of points in a metric ball. These problems have known polynomial time
algorithms (for instance, the $k$\th distance may be found by applying
a selection algorithm to the coefficients of the distance matrix);
however, we are interested in algorithms that scale well to large data
sets, so we seek algorithms that take subquadratic time (substantially
smaller than the time to list all distances). To achieve this, we
require the metric space to be defined \emph{implicitly}, for instance
as the distances in a sparse weighted graph or as the distances among
points in a Euclidean metric space. Despite the increased difficulty
of working with implicit metrics, we show that the problems we study
can be solved efficiently.

\subsection{Greedy permutation}
\seclab{def} In the first sections of this paper we are interested in
an ordering problem on metric spaces: the construction of greedy
permutations. We solve this problem exactly and approximately, for the
shortest path metrics of sparse weighted graphs and for
high-dimensional Euclidean spaces.

\begin{figure}[t]
    \centering\includegraphics[height=1.25in]{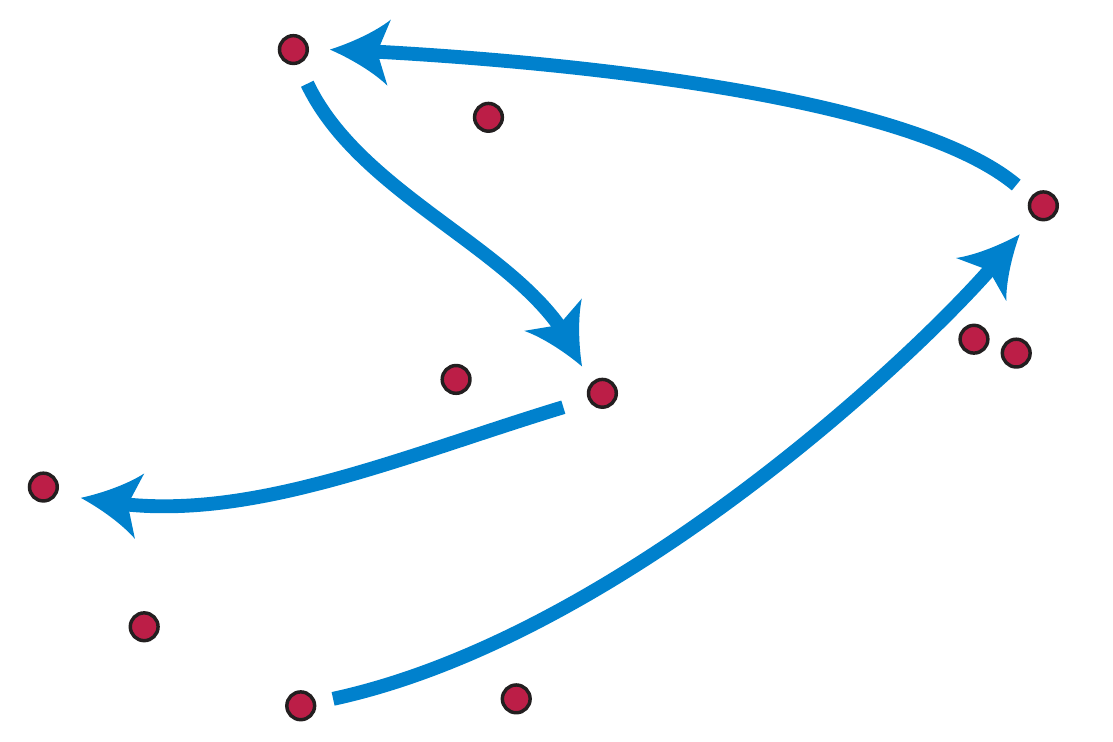}%
    \qquad%
    \qquad%
    \qquad%
    \includegraphics[height=1.5in]{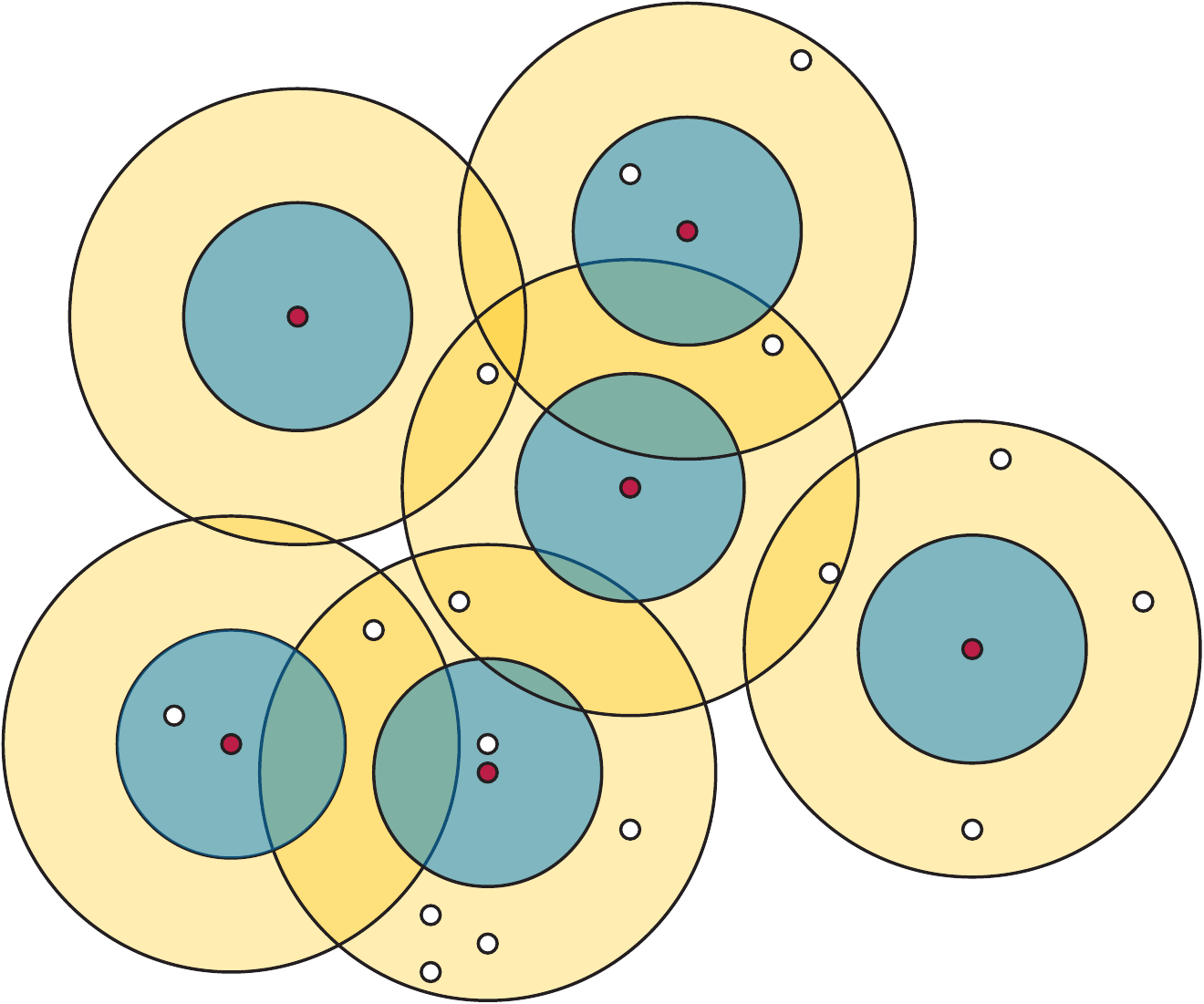}%
    \caption{Left: The first five points of a greedy permutation. Each
       point is as far as possible from all previously chosen
       points. Right: The six red points at the centers of the disks
       form an $r$-net of the set $V$ of red and white points, where
       $r$ is the radius of the large yellow disks. The darker disks
       of radius $r/2$ are disjoint from each other, and the disks of
       radius $r$ together cover all of~$V$.}
    \figlab{defs}
\end{figure}

\newcommand{\XFrame}[2]{%
   \begin{minipage}{0.24\linewidth}%
       \begin{minipage}{0.99\linewidth}%
           \fbox{\includegraphics[page=#1,width=0.97\linewidth]{figs/greedy}}
       \end{minipage}\\[-0.6cm]%
       $~$\;\;\;\;{#2}
   \end{minipage}
}

\begin{figure}[t]
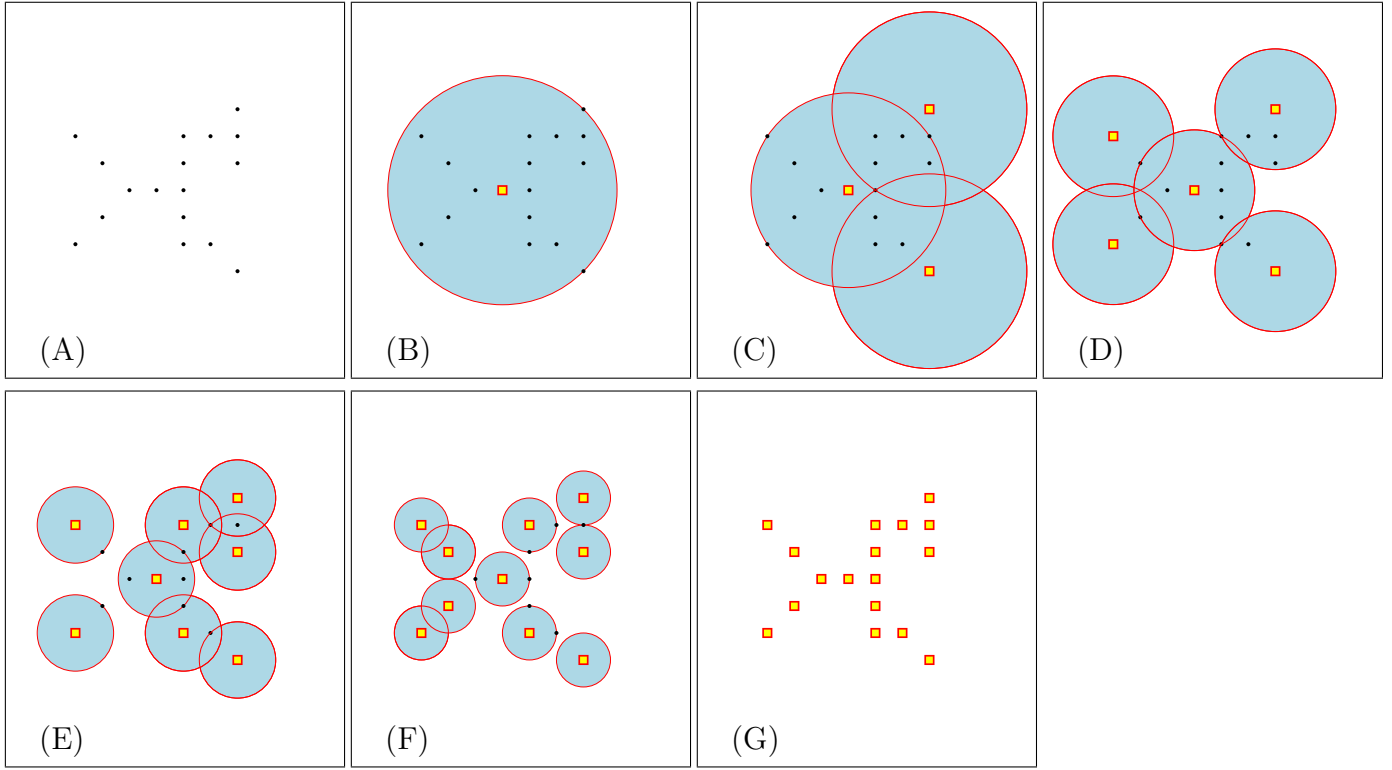

    \XFrame{1}{(A)}%
    \XFrame{2}{(B)}%
    \XFrame{3}{(C)}%
    \XFrame{4}{(D)}%
    
    ~\\[-0.28cm]
    
    
    \XFrame{5}{(E)}%
    \XFrame{6}{(F)}%
    \XFrame{7}{(G)}
    
    \caption{(A) A point set. (B)--(G) The different representations
       of this point set as a union of balls, as provided by different
       prefixes of the greedy permutation. As this demonstrates, one
       can think about a prefix of the greedy permutation as a partial
       representation of the point set that keeps improving as the
       prefix used gets longer.  }
    \figlab{multi:resolution}
\end{figure}

A permutation $\Pi = \permut{\pi_1, \pi_2, \ldots}$ of the vertices of
a metric space $(V, \metricC)$ is a \emphi{greedy permutation} (also
called a \emph{farthest-first traversal} or \emph{farthest point
   sampling}) if each vertex $\pi_i$ is the farthest in $V$ from the
set $\Pi_{i-1} = \brc{\pi_1, \ldots, \pi_{i-1}}$ of preceding vertices
(\figref{defs}, left).  Greedy permutations were introduced by
Rosenkrantz \etal{}~\cite{rsl-ashtsp-77} for the ``farthest
insertion'' traveling salesman heuristic, and used by
Gonzalez~\cite{g-cmmid-85} to 2-approximate the $k$-center.  Different
prefixes of the greedy permutation provide different multi-resolution
clusterings of the input point set; see \figrefpage{multi:resolution}.

Greedy permutations are closely related to another concept for finite
metric spaces, $r$-nets.  An \emphi{$r$-net} for the metric space
$(V, \metricC)$ and the numerical parameter $r$ is a subset $\net$ of
the points of $V$ such that no two of the points of $\net$ are within
distance $r$ of each other, and such that every point of $V$ is within
distance $r$ of a point of $\net$.  Equivalently, the closed
$r/2$-balls centered at the points of $\net$ are disjoint, and the
closed $r$-balls around the same points cover all of $V$
(\figref{defs}, right). Each prefix of a greedy permutation is an
$r$-net, for $r$ equal to the minimum distance between points in the
prefix, and for every $r$ an $r$-net may be obtained as a prefix of a
greedy permutation%
\footnote{The notion of nets is closely related to congruent disc
   packing, which was studied in the end of the 19\th century by Thue
   (see \cite[Chapter 3]{pa-cg-95}). It is however natural to assume
   that the concept is much older, as it is related to numerical
   integration and discrepancy \cite{m-gd-99}.}.

Greedy permutations may be computed for metric spaces in $O(n^2)$
time, and for graphs in the same time as all pairs shortest paths, by
a naive algorithm \InFullVer{(\secref{naive})} that maintains the
distances of all points from the selected points. The only previous
improvement on the naive algorithm, by Har-Peled and
Mendel~\cite{hm-fcnld-06} defines a concept of approximation for
greedy permutations that we will also use. They showed that
\emphi{$(1+\eps)$-greedy permutations} can be computed in
$O(n \log n)$ time in metric spaces with constant doubling dimension;
these are permutations $\Pi = \permut{\pi_1, \pi_2, \ldots}$ for which
there exists a sequence of numbers $r_1 \geq r_2 \geq \cdots$ such
that
\begin{compactenum}[\qquad(A)]
    \item the maximum distance of a point of $V$ from $\Pi_i$ is in
    the range $\pbrc[]{r_i, \, (1+\eps)r_i}$, and
    \item the distance between every two points $u,v \in \Pi_i$ is at
    least $r_i$.
\end{compactenum}
In this paper we give approximation schemes for metric spaces defined
by sparse graphs, and high-dimensional Euclidean spaces, neither of
which have constant doubling dimension.  Greedy permutations for graph
distances were previously mentioned by Gu \etal \cite{gjg-kacdmd-11},
in connection with an application in molecular dynamics, but rejected
by them because the naive algorithm was too slow.

One reason for interest in greedy or $(1+\eps)$-greedy permutations is
that, in a single structure, they approximate an optimal clustering
for all possible resolutions. Specifically, the prefix of the first
$k$ vertices in such a permutation, provides, for any $k$, a
$2(1+\eps)$-approximation to the optimal $k$-center clustering of
$\Graph$.  \InFullVer{ (See \lemref{g:to:k:center} for an easy proof
   of this.)}  The $k$-center problem may be 2-approximated in $O(kn)$
time by computing the first $k$ vertices of a greedy permutation, and
is NP-hard to approximate to a ratio better
than~2~\cite{g-cmmid-85}.\InFullVer{\footnote{Another 2-approximation
      for the $k$-center by Hochbaum and Shmoys~\cite{hs-mor-85} is
      often erroneously credited as being the origin of the
      farthest-first traversal method, but actually uses a different
      algorithm.}} {} For points in Euclidean spaces of bounded
dimension a linear-time 2-approximation is
known~\cite{h-cm-04,hr-npltaedp-13}.  Thorup \cite{t-qmcflsg-05}
provided a fast $k$-center approximation for graph shortest path
metrics with a single choice of~$k$. The $k$-center problem may be
solved exactly on trees and cactus graphs in $O(n\log n)$
time~\cite{fj-fkppc-83}. Voevodski \etal \cite{vbrtx-ecldi-10} use an
algorithm closely related to greedy permutation to approximate a
different clustering problem, $k$-medians.

Intuitively, every prefix of a greedy permutation is as informative as
possible about the whole set, so greedy permutations form a natural
ordering in which to stream large data sets. Because of these
properties, greedy permutations have many additional applications,
including color quantization~\cite{x-ciqmmid-97}, progressive image
sampling~\cite{elpz-fpspis-97}, selecting landmarks of probabilistic
roadmaps for motion planning~\cite{mab-aca-98}, point cloud
simplification~\cite{md-npcsa-03}, halftone mask
generation~\cite{smr-dfpmih-04}, hierarchical
clustering~\cite{dl-pghc-05}, detecting isometries between surface
meshes~\cite{lf-mvsc-09}, novelty detection and time management for
autonomous robot exploration~\cite{ggd-aaueotm-12}, industrial fault
detection~\cite{aye-inpsf-12}, and range queries seeking diverse sets
of points in query regions~\cite{aaimv-dnpp-13}.

\subsection{Distance selection and approximate range %
   counting}
\seclab{dist:count:select}

The \emph{distance selection} problem, in computational geometry, has
as input a set of points in $\Re^d$; the output is the $k$\th smallest
distance defined by a pair of points of $\PntSet$.  It is believed
that such exact distance selection requires $\Omega\pth{n^{4/3}}$ time
in the worst case \cite{e-rcsgp-95}, even in the plane (in higher
dimensions the bound deteriorates). Recently, Har-Peled and Raichel
\cite{hr-npltaedp-13} provided an algorithm that
$(1+\eps)$-approximates this distance in $O(n/\eps^d)$ time.

We are interested in solving the problem for the finite metric case.
Specifically, consider a shortest path metric defined over a graph
$\Graph$ with $n$ vertices and $m$ edges.  Given
$k \in \IntRange{\binom{n}{2}} = \brc{1, \ldots, \binom{n}{2}}$, we
would like to compute the $k$\th smallest distance in this shortest
path metric. This problem was studied for trees \cite{mtzc-aklpt-81},
where Frederickson and Johnson \cite{fj-fkppc-83} provided a beautiful
algorithm that works by using tree separators, and selection in sorted
matrices \cite{fj-gsrsm-84}.

The ``dual'' problem to distance selection is \emph{distance
   counting}. Here, given a distance $r$, the task is to count the
number of pairs of points of $\PntSet$ that are of distance $\leq r$.
While the problems are essentially equivalent in the exact case,
approximate distance counting seems to be significantly easier than
selection.

Throughout this paper, $\Graph = (\Vertices, \Edges)$ will denote an
undirected graph with $n$ vertices and $m$ edges, with non-negative
edge weights that obey the triangle inequality.  The shortest path
distances in $\Graph$ induce a metric $\metricC$. Specifically, for
any $u,v \in \Vertices$, let $\distG{u}{v}$ denotes the shortest path
between $u$ and $v$ in $\Graph$.
Given a graph $\Graph=(\Vertices, \Edges)$, and a query distance $r$,
the \emphi{set of $r$-short} pairs is
\begin{align}
    \eqlab{short:pairs}%
    \ShortPairs{r}%
    =%
    \Set{ \Bigl. \brc{u,v} \subseteq \Vertices}{u \neq v \text{ and }
       \distG{u}{v} \leq r}.
\end{align}
In the \emphi{distance counting} problem, the task is to compute (or
approximate) $\cardin{P_{\leq r}}$. In the \emphi{distance selection}
problem, given $k \in \IntRange{\binom{n}{2}}$, the task is to compute
(or approximate) the smallest $r$ such that
$\cardin{\ShortPairs{r}} \geq k$.

Distance counting is easy to approximate using known techniques, since
this problem is malleable to random sampling, see \cite{c-adsrh-14}
and references therein. However, approximate distance selection is
significantly \emph{harder}, as random sampling can not be used in
this case -- indeed, trying to use approximate distance counting (or
sketches approach as in \cite{c-adsrh-14}), may result in an
arbitrarily bad approximation to the $k$\th distance, if the
$(k-1)$\th and $(k+1)$\th distance are significantly smaller and
larger, respectively, than the $k$\th distance (or similar sparse
scenarios).

\subsection{New results}

\paragraph{Greedy permutation for sparse graphs.} %
In \secref{g:sparse}, we show that an $(1+\eps)$-greedy permutation
can be found for graphs with $n$ vertices and $m$ edges in time
$O \pth{ \eps^{-1} m \log n \log (n/\eps) } = \Otilde
\pth{m}$\footnote{The
   $\Otilde$ notation hides logarithmic factors in $n$ and polynomial
   terms in $1/\eps$. We assume throughout the paper that
   $n = O(m)$.}.
        
\paragraph{Approximate greedy permutation for high dimensional
   Euclidean space.} %

In \secref{g:euclidean}, we show that an approximate greedy
permutation can be computed for a set of points in high-dimensional
Euclidean space.  The algorithm runs in subquadratic time and has
polynomial dependency on the dimension.  Our approximations are based
on finding $r$-nets (or in the Euclidean case approximate $r$-nets)
for a geometric sequence of values of~$r$%

In an earlier paper \cite{his-eshd-13}, the authors showed that for
high dimensional point sets one can get a sparse spanner.  Applying
the above algorithm for sparse graphs to this spanner yields a greedy
permutation, but with significantly weaker bounds, as the stretch in
the constructed spanner is at lease $2$.

\paragraph{Exact greedy permutation for bounded tree-width.} %
In \secref{e:g:treewidth}, \InShortVer{with the details in the full
   version of this paper, } we show how to find an exact greedy
permutation for graphs of bounded treewidth, in time
$\Otilde \pth{n^{3/2}}$, by partitioning the input graph into small
subgraphs separated from the rest of the graph by $O(1)$ vertices, and
by using an orthogonal range searching data structure in each subgraph
to find the farthest vertex from the already-selected vertices.

\paragraph{Distance selection in planar graphs.}
In \secref{planar:counting}, we show how to $O(1)$-approximate the
$k$\th distance in a planar graph $\Graph$. Specifically, given $k$,
the algorithm computes in near linear time, a number $\alpha$, such
that the $k$\th shortest distance in $\Graph$ is at least $\alpha$,
and at most $O(\alpha)$.  This algorithm uses a planar separator and
distance oracles in an interesting way to count distances.

\InShortVer{%
   \medskip\noindent For space reasons we defer many proofs to the
   full version of this paper.%
}

\section{Approximate greedy permutation on a %
   sparse graph}
\seclab{g:sparse}%
\seclab{naive}

We are interested in approximating the greedy permutation for a graph
$\Graph$.  Among other motivations, this provides a good approximation
for $k$-center clustering:

\begin{lemma}
    \lemlab{g:to:k:center}%
    If $\Pi$ is a $(1+\eps)$-greedy permutation of
    $\Mtr = (\Vertices,\metricC)$, then, for all $k$, $\Pi_k$ provides
    a $2(1+\eps)$-approximation to the optimal $k$-center clustering
    and minimax diameter $k$-clustering of $\Mtr$.
\end{lemma}

   \begin{proof}
       By Property (A) of such a permutation (see \secref{def}), all
       points of $\Vertices$ can be covered by balls of radius
       $(1+\eps)r_k$ centered at $\pi_1, \ldots, \pi_k$; these balls
       have diameter $\leq 2(1+\eps)r_k$. Let
       $S = \Pi_{k} \cup \brc{v}$, where $v$ is the farthest point in
       $\Vertices$ from $\Pi_k$. By the definition of $r_i$ and by
       Property (B) of these permutations, every two points in $S$
       have distance at least $r_i$, so no $k$ clusters of radius
       smaller than $r_i/2$ or diameter smaller than $r_i$ can cover
       the $k+1$ points in~$S$.
   \end{proof}
   
   A naive algorithm for computing the greedy permutation maintains
   for each vertex $v$ its distance $\distC{v}$ to the set of centers
   picked so far, and uses these distances as priorities in a
   max-heap, which it uses to select each successive center, using
   Dijkstra's algorithm to update the distances after each center is
   picked. There are $n$ instantiations of Dijkstra's algorithm,
   taking time $O( n(m+n \log n))$.

   To improve performance, we may avoid adding a vertex $v$ to the
   min-heap used within Dijkstra's algorithm unless its tentative
   distance is smaller than $\distC{v}$, preventing the expansion of
   vertices for which the distance from $v_i$ is no smaller than
   $\distC{v}$. This idea does not immediately improve the worst-case
   running time of the algorithm but will be important in our
   approximation algorithm.%

\subsection{Computing an $r$-net in a sparse graph}
\seclab{net}

We compute an $r$-net in a sparse graph using a variant of Dijkstra's
algorithm with the sequence of starting vertices chosen in a random
permutation. A similar idea was used by Mendel and Schwob
\cite{ms-fckrp-09} for a different problem; however, using this method
for our problem involves a more complicated analysis.

Let $\Graph = (\Vertices, \Edges)$ be a weighted graph with $n$
vertices and $m$ edges, let $r >0$, and let $\pi_i$ be the $i$\th
vertex in a random permutation of~$\Vertices$. For each vertex $v$ we
initialize $\delta(v)$ to~$+\infty$. In the $i$\th iteration, we test
whether $\delta(\pi_i) \geq r$, and if so we do the following steps:
\smallskip
\begin{compactenum}[\quad 1.]
    \item Add $\pi_i$ to the resulting net $\net$.
    \item Set $\delta(\pi_i)$ to zero.
    \item Perform Dijkstra's algorithm starting from $\pi_i$, modified
    \InFullVer{as in \secref{naive}}%
    to avoid adding a vertex $u$ to the priority queue unless its
    tentative distance is smaller than the current value of
    $\delta(u)$. When such a vertex $u$ is expanded, we set
    $\delta(u)$ to be its computed distance from $\pi_i$, and relax
    the edges adjacent to $u$ in the graph.
\end{compactenum}
\smallskip
\noindent%
The difference from the algorithm of Mendel and Schwob is that their
algorithm initiates an instance of Dijkstra's algorithm starting from
every vertex $\pi_i$, whereas we do so only when
$\delta(\pi_i)\geq r$.

\begin{lemma}
    The set $\net$ is an $r$-net in $\Graph$.
\end{lemma}

\InFullVer{%
   \begin{proof}
       By the end of the algorithm, each $v \in \Vertices$ has
       $\delta(v) <r$, for $\delta(v)$ is monotonically decreasing,
       and if it were larger than $r$ when $v$ was visited then $v$
       would have been added to the net.
    
       An induction shows that if $\ell = \delta(v)$, for some vertex
       $v$, then the distance of $v$ to the set $\net$ is at most
       $\ell$.  Indeed, for the sake of contradiction, let $j$ be the
       (end of) the first iteration where this claim is false. It must
       be that $\pi_j \in \net$, and it is the nearest vertex in
       $\net$ to $v$. But then, consider the shortest path between
       $\pi_j$ and $v$. The modified Dijkstra must have visited all
       the vertices on this path, thus computing $\delta(v)$ correctly
       at this iteration, which is a contradiction.
    
       Finally, observe that every two points in $\net$ have distance
       $\geq r$. Indeed, when the algorithm handles vertex
       $v\in \net$, its distance from all the vertices currently in
       $\net$ is $\geq r$, implying the claim.
   \end{proof}
}

\begin{lemma}%
    \lemlab{update:delta}%
    Consider an execution of the algorithm, and any vertex
    $v \in \Vertices$. The expected number of times the algorithm
    updates the value of $\delta(v)$ during its execution is
    $O( \log n )$, and more strongly the number of updates is
    $O(\log n)$ with high probability.
\end{lemma}

\InFullVer{%
   \begin{proof}
       For simplicity of exposition, assume all distances in $\Graph$
       are distinct.  Let $S_i$ be the set of all the vertices
       $x \in \Vertices$, such that the following two properties both
       hold: \smallskip %
       \begin{compactenum}[\qquad(A)]
           \item $\dist{x}{v} < \dist{v}{\Pi_i}$, where
           $\Pi_i = \brc{\pi_1,\ldots, \pi_i}$.
           \item If $\pi_{i+1}=x$ then $\delta(v)$ would change in the
           $(i+1)$\th iteration.
       \end{compactenum}
       \medskip%
       Let $s_i = \cardin{S_i}$.  Observe that
       $S_1 \supseteq S_2 \supseteq \cdots \supseteq S_n$, and
       $\cardin{S_n} = 0$.
    
       In particular, let $\Event_{i+1}$ be the event that $\delta(v)$
       changed in iteration $(i+1)$ -- we will refer to such an
       iteration as being \emphi{active}. If iteration $(i+1)$ is
       active then one of the points of $S_i$ is $\pi_{i+1}$. However,
       $\pi_{i+1}$ has a uniform distribution over the vertices of
       $S_{i}$, and in particular, if $\Event_{i+1}$ happens then
       $s_{i+1} \leq s_i/2$, with probability at least half, and we
       will refer to such an iteration as being \emphi{lucky}. (It is
       possible that $s_{i+1} < s_i$ even if $\Event_{i+1}$ does not
       happen, but this is only to our benefit.)  After $O( \log n)$
       lucky iterations the set $S_i$ is empty, and we are
       done. Clearly, if both the $i$\th and $j$\th iteration are
       active, the events that they are each lucky are independent of
       each other. By the Chernoff inequality, after $c \log n$ active
       iterations, at least $\ceil{ \log_2 n }$ iterations were lucky
       with high probability, implying the claim. Here $c$ is a
       sufficiently large constant.
   \end{proof}

   Interestingly, in the above proof, all we used was the monotonicity
   of the sets $S_1, \ldots, S_n$, and the fact that if $\delta(v)$
   changes in an iteration then the size of the set $S_i$ shrinks by a
   constant factor with good probability in this iteration. This
   implies that there is some flexibility in deciding whether or not
   to initiate Dijkstra's algorithm from each vertex of the
   permutation, without damaging the number of times of the values of
   $\delta(v)$ are updated. We will use this flexibility later on. %
}%

\begin{lemma}%
    \lemlab{net:alg}%
    Given a graph $\Graph = (\Vertices, \Edges)$, with $n$ vertices
    and $m$ edges, the above algorithm computes an $r$-net of $\Graph$
    in $O( (n+m)\log n)$ expected time.
\end{lemma}

   \begin{proof}
       By \lemref{update:delta}, the two $\delta$ values associated
       with the endpoints of an edge get updated $O( \log n)$ times,
       in expectation, during the algorithm's execution. As such, a
       single edge creates $O(\log n)$ decrease-key operations in the
       heap maintained by the algorithm. Each such operation takes
       constant time if we use Fibonacci heaps to implement the
       algorithm.
   \end{proof}%

\subsection{An approximation whose time depends on %
   the spread}

Given a finite metric space $(V, \metricC)$ defined over a set $V$,
its \emphi{spread} is the ratio between the maximum and minimum
distance in the metric; formally, %
   \begin{align*}
       \SpreadX{\metricC}%
       =%
       \max_{{u, v \in V, u \neq v }} \distG{u}{v}%
       /%
       \min_{{u, v \in V, u \neq v }} \distG{u}{v}.
   \end{align*}%

Let graph $\Graph=(\Vertices,\Edges)$ and $\eps > 0$ be given.  Assume
for now that the minimum edge length is $1$, and that the diameter of
$\Graph$ is at most $\Diameter$. Set
$r_i = \Diameter / (1+\eps)^{i-1}$, for
$i=1,\ldots, M= \ceil{\log_{1+\eps} \Delta}$.  We compute a sequence
of nets in a sequence of iterations.  In the first iteration, compute
an $r_1$-net $\net_1$ of $\Graph$, using \lemref{net:alg}. In the
beginning of the $i$\th iteration, for $i > 1$, let
$\Snet_i = \cup_{j < i} \net_i$. Using Dijkstra, mark as \emphi{used}
all vertices within distance $r_i$ of~$\Snet_i$. Compute an $r_i$-net
$\net_i$ in $\Graph$, modifying the algorithm of \lemref{net:alg} by
disallowing used vertices from being considered as net points.

After completing these computations, combine the vertices into a
single permutation in which the vertices of $\net_i$ form the $i$\th
contiguous block. Within this block, the ordering of the vertices of
$\net_i$ is arbitrary. The following is easy to verify, and we omit
the easy proof.

\begin{lemma}
    Let $\net = \net_1 \cup \cdots \cup \net_t$, for an arbitrary
    $t$. Then the distance of every vertex $v \in \Graph$ from $\net$
    is at most $r_t$, and the distance between any pair of vertices of
    the net is at least $r_t$.
\end{lemma}

That is, the net computed by the $t$\th iteration is a ``perfect''
$r_t$ net. In between such blocks, it might be less then perfect. %
   Formally, we have the following (again, we omit the relatively easy
   proof).
   \begin{lemma}
       Let $\pi$ be the permutation computed by the above algorithm,
       and consider the $i$\th vertex $\pi_i$ in this
       permutation. Assume that $\pi_i \in \net_t$. Then we have the
       following guarantees:
       \begin{compactenum}[\qquad (A)]
           \item The distance between any two vertices of
           $\Pi_i = \brc{\pi_1, \ldots, \pi_i}$ is at least $r_t$.
           
           \item The distance of any vertex of $v \in \Vertices$ from
           $\Pi_i$, is at most $r_{t-1} = (1+\eps)r_t$.
       \end{compactenum}
   \end{lemma}%

\begin{lemma}%
    \lemlab{greedy_spread}%
    Let $\Graph=(\Vertices,\Edges)$ be a graph, let $\eps > 0$, and
    let $\Spread$ be the spread of $\Graph$.  Then, one can compute a
    $(1+\eps)$-greedy permutation in
    $O\pth{ \eps^{-1} \pth{ n + m } \log n \log \Spread }$ expected
    time.
\end{lemma}

\begin{proof}
    A $2$-approximation to the diameter of $\Graph$ can be found by
    running Dijkstra's algorithm from an arbitrary starting
    vertex. The minimum distance in $\Graph$ is achieved by an edge
    (since the edge weights are positive), so it can be computed in
    linear time. After scaling, we can use the above algorithm, with
    $M = O\pth{\log_{1+\eps} \Spread} = O\pth{ \eps^{-1} \log
       \Spread}$
    iterations, each using a modified version of the algorithm of
    \secref{net}. It is straightforward to modify the analysis to show
    that the each such iteration takes $O( (n + m) \log n )$ time in
    expectation.
\end{proof}

\subsection{Eliminating the dependence on the spread}
\seclab{elim-spread}

An arbitrary graph $\Graph$ may not have small enough spread to apply
the previous algorithm directly. In this case, following by-now
standard methods for eliminating the dependence on spread (see Section
4 of~\cite{ms-fckrp-09}), we simulate the algorithm more efficiently,
using a value of $\eps$ smaller by a constant factor to make up for
some additional approximation in our simulation.

Consider an iteration of the above algorithm for distance $r_i$. Edges
longer than $n r_i$ can be ignored or (conceptually) deleted, as they
cannot be used by the $r$-net algorithm of \secref{net}. Similarly,
edges of length $O\pth{\eps r_i /n^2}$ can be collapsed and treated as
having length zero. Thus, an edge $e$ of length $\ell$ is active when
$r_i$ is in the interval $\pbrc{\frac{c\eps \ell}{n^2}, \ell n}$,
which happens for
$O\pth{\log_{1+\eps}(n^3/\eps)}=O\pth{\eps^{-1}\log(n/\eps)}$
iterations.  Let $m_i$ be the number of active edges in the $i$\th
iteration, and let $\Graph_i$ be the resulting graph, in which all the
edges of lengths $\leq \eps r_i/n^2$ are contracted (the resulting
super vertex is identified with one of the original vertices), and all
the edges of length $>r_i n$ are removed. Any singleton vertex in this
graph is not relevant for computing the permutation in this
resolution, and it can be ignored. The running time of
\lemref{net:alg} on $\Graph_i$ is $O\pth{ m_i \log m_i }$.

When the algorithm moves to the next iteration, it needs to introduce
into $\Graph_i$ all the new edges that become active. Using a careful
implementation, this can be done in $O( 1 )$ amortized time, for any
newly introduced edge. Similarly, edges that become inactive should be
deleted.  Of course, if there are no active edges, the algorithm can
skip directly to the next resolution. This can be easily done, by
putting the edges into a heap, sorted by their length, and adding the
edges and removing them as the algorithm progresses down the
resolutions.

The overall expected running time of this algorithm is
$O\pth{\sum_i m_i \log m_i + m \log m}$. However, since every edge is
active in $O\pth{ \eps^{-1} \log(n/\eps)}$ iterations, we get that the
expected running time is $O\pth{ \eps^{-1} m \log n\log(n/\eps) }$. We
thus get the following.

\begin{theorem}
    \thmlab{greedy-graph}%
    Given a non-negatively weighted graph
    $\Graph = (\Vertices, \Edges)$, with $n$ vertices and $m$ edges,
    and a parameter $\eps > 0$, one can compute a
    $(1+\eps)$-approximate greedy permutation for $\Graph$ in expected
    time $O\pth{ \eps^{-1} m \log n\log(n/\eps) }$.
\end{theorem}

\InFullVer{

   \subsection{$k$-center clustering for bounded spread %
      with integer weights}

   Our greedy permutation algorithm for sparse graphs leads to a fast
   $(2+\epsilon)$-approximation to the $k$-center problem for graph
   metrics. In this case, it is possible to eliminate the dependence
   on $\epsilon$, giving a $2$-approximation (best possible as
   achieving a smaller approximation ratio is NP-hard).

\begin{theorem}%
    \thmlab{k:center:intger:spread}%
    Let $\Graph$ be a graph with $n$ vertices and $m$ edges, with
    positive integer weights on the edges, and spread $\Spread$. Given
    $k$, one can compute a $2$-approximation to the optimal $k$-center
    clustering of $\Graph$, in $O\pth{ m \log n \log \Spread}$ time
    both in expectation and with high probability.
\end{theorem}

\begin{proof}
    Using Dijkstra's algorithm starting from an arbitrary vertex in
    $\Graph$, compute, in $O(m + n \log n)$ time, a number $\Delta$,
    such that
    $\DiameterX{\Graph} \leq \Delta \leq 2\DiameterX{\Graph}$.  We
    next perform a binary search for the radius $\ropt$ of the optimal
    $k$-center clustering of $\Graph$, in the range
    $1, \ldots, \Delta$.
    
    Given a candidate radius $x$, let $r=2x$ and compute an $r$-net of
    $\Graph$ using the algorithm of \lemref{net:alg}. It is easy to
    verify that if $x \geq \ropt$, then the resulting net $\net$ has
    at most $k$ vertices in it. Indeed, consider all the vertices in
    $\Graph$ assigned to a single cluster $C$ in the optimal
    $k$-center clustering, and observe that
    $\DiameterX{C} \leq 2\ropt \leq r$. Therefore, at most one vertex
    of $C$ may belong to any $r$-net, so every $r$-net for this value
    of $r$ has at most $k$ vertices.  On the other hand, if $x$ is too
    small then the resulting $r$-net has more then $k$ vertices.
    
    Thus, using the $r$-net procedure of \lemref{net:alg} as a decider
    in a binary search yields the desired approximation algorithm.
\end{proof}%
}

\section{Approximate greedy permutation on %
   Euclidean metrics}
\seclab{g:euclidean}

\subsection{Approximate nets}

\begin{lemma}[Johnson--Lindenstrauss lemma \cite{jl-elmih-84}]%
    \lemlab{JL}%
    For any $\eps>0$, every set of $n$ points in Euclidean space
    admits an embedding into
    $(\Re^{O(\log n / \eps^2)}, \|\cdot \|_2)$, with distortion
    $1+\eps$.
\end{lemma}

\begin{defn}
    Let $H$ be a family of hash functions mapping $\Re^d$ to some
    universe $U$.  We say that $H$ is
    \emphi{$(\delta,c \delta, p_1, p_2)$-sensitive} if for any
    $x,y\in \Re^d$ it satisfies the following properties:
    \begin{compactenum}[\qquad(A)]
        \item If $\|x-y\|_2\leq \delta$ then
        $\Pr_{h\in H}[h(x)=h(y)]\geq p_1$.
        \item If $\|x-y\|_2\geq c \delta$ then
        $\Pr_{h\in H}[h(x)=h(y)]\leq p_2$.
    \end{compactenum}
\end{defn}

\begin{lemma}[Andoni \& Indyk \cite{ai-nohaa-06}]%
    \lemlab{LSH}%
    For any $\delta>0$, dimension $d>0$, and $c>1$, there exists a
    $\pth{\delta, c \delta, 1/n^{1/c^2 + o(1)}, 1/n}$-sensitive family
    of hash functions for $\Re^d$, where every function can be
    evaluated in time $O\pth{d\cdot n^{o(1)}}$.
\end{lemma}

We extend the notion of nets to \emphi{$c$-approximate $r$-nets}, for
$c\geq 1$, $r>0$, and metric space $(V, \metricC)$. These are subsets
$\net$ of the points of $V$ such that no two of the points of $\net$
are within distance $r$ of each other, and such that every point of
$V$ is within distance $c\cdot r$ of a point of $\net$.

\begin{theorem}%
    \thmlab{high-dim-approx-net}%
    Let $d>0$, $r>0$, $\eps>0$.  Given a set $X$ of $n$ points in
    $\Re^d$, one can compute in expected running time
    $O\pth{\eps^{-2} n^{1+1/(1+\eps)^2+o(1)}}$ a set $\net\subseteq X$
    such that $\net$ is a $(1+\eps)$-approximate $r$-net for the
    Euclidean metric on $X$ with high probability.
\end{theorem}

\begin{proof}
    By \lemref{JL} we may assume that $d = O\pth{\log n / \eps^2}$,
    since otherwise we can embed $X$ into Euclidean space of dimension
    $O\pth{\log n / \eps^2}$, with distortion $1+O(\eps)$. Any
    $\pth{1+O(\eps)}$-approximate $r$-net for this new point set is a
    $\pth{1+O(\eps)}$-approximate $r$-net for $X$.
    
    Let $c=1+\eps$.  Let $H$ be the $(r,cr,p_1,p_2)$-sensitive family
    of hash functions given by \lemref{LSH}, with
    $p_1=1/n^{1/c^2+o(1)}$, $p_2=1/n$.  We sample
    $k=O\pth{(1/p_1) \cdot \log n} = O\pth{n^{1/c^2+o(1)}}$ hash
    functions $h_1,\ldots,h_k\in H$.  For every $i\in \{1,\ldots,k\}$,
    and for every $x\in X$, we evaluate $h_i(x)$.
    
    We construct a set $\net \subseteq X$ which is initially empty,
    and it will be the desired net at the end of the algorithm.
    Initially, we consider all points in $X$ as being unmarked.  We
    pick an arbitrary ordering $x_1,\ldots,x_n$ of the points in $X$,
    and we iterate over all points in this order.  When the iteration
    reaches point $x_i$, if it is already marked, we skip it, and
    continue with the next point.  Otherwise, we add $x_i$ to $\net$,
    and we proceed as follows.  Let $M_i$ be the set of all currently
    marked points.  Let
    $S_i = \bigcup_{j=1}^k h_j^{-1}(x_i) \setminus M_i$ be the set of
    unmarked points that are hashed to the same value in at least one
    of the sampled hash functions $h_1,\ldots,h_k$.  We mark all
    points $y\in S_i$, such that $\|x_i-y\|_2 \leq c \cdot r$.  This
    completes the construction of the desired set $\net$.
    
    \InFullVer{%
       We next argue that $\net$ is indeed a $c$-approximate $r$-net.
       Every point $y\in X\setminus \net$ must have been marked when
       considering some earlier point $x_i \in \net$, implying that
       $\|x_i-y\|_2 \leq c\cdot r$. Thus, every non-net point is
       covered by a net point.  On the other hand, consider a pair of
       points $x_i,x_j$ ($i<j$) for which $\|x_i,x_j\|_2 \leq r$.
       Then with high probability, there exists
       $t'\in \{1,\ldots,k\}$, with $h_{t'}(x_i)=h_{t'}(x_j)$.  If so,
       $x_j$ will be marked when we consider $x_i$, preventing it from
       belonging to $\net$.  Therefore, with high probability, we have
       that for any $x,x'\in \net$, $\|x-x'\|_2 > r$.  This
       establishes that $\net$ is indeed a $c$-approximate $r$-net.
    
       It remains to bound the running time.  We perform $n\cdot k$
       hash function evaluations in time
       $O\pth{d \cdot n^{o(1)}} = O\pth{\eps^{-2} n^{o(1)}}$ each,
       totaling time $O\pth{\eps^{-2} n^{1+1/c^2+o(1)}}$.  The
       remaining time is dominated by $O\pth{\sum_{i=1}^n |S_i|}$.
       Each point is marked at most once, so
       $O\pth{\sum_{i=1}^n |S_i|} = O(n+L)$, where $L$ is the total
       number of \emph{false positives}, i.e.~the number of triples
       $x,y,t$ with $x,y\in X$, $\|x-y\|_2 > c\cdot r$,
       $t\in \{1,\ldots,k\}$, and $h_t(x)=h_t(y)$.  For any
       $x,y\in X$, with $\|x-y\|_2 > c \cdot r$, and for any
       $t\in \{1,\ldots,k\}$, we have $\Pr[h_t(x)=h_t(y)] \leq 1/n$.
       Since there are $O\pth{n^2}$ pairs of points, we conclude that
       the expected number of false positives is
       $\mathbf{E}[L] \leq O\pth{\sum_{x\neq y \in X} \sum_{t=1}^{k}
          1/n} = O\pth{n^{1+1/c^2+o(1)}}$.
       We conclude that the total expected running time of the
       algorithm is $O\pth{\eps^{-2} n^{1+1/c^2+o(1)}}$, as required.%
    }%
    \InShortVer{%
       For the proof that this is indeed a $c$-approximate $r$-net,
       and the time analysis, see the full version of this paper.%
    }
\end{proof}

\subsection{An approximation whose time %
   depends on the spread}

\begin{lemma}
    Let $d\geq 1$, let $X$ be a set of $n$ points in $\Re^d$, let
    $\eps > 0$, and let $\Spread$ be the spread of the Euclidean
    metric on $X$.  Then, one can compute in
    $O\pth{\eps^{-2} n^{1+1/(1+\eps)^2 + o(1)} \log \Spread}$ expected
    time a sequence that is a $(1+\eps)$-greedy permutation for the
    Euclidean metric on $X$, with high probability.
\end{lemma}

\begin{proof}
    We use the algorithm from \lemref{greedy_spread}.  A 2-approximate
    diameter can easily be computed in linear time, by choosing one
    point arbitrarily and finding a second point as far from it as
    possible.  The only new needed observation is that it is
    sufficient for the algorithm to compute
    $\pth{1+O(\eps)}$-approximate $r$-nets, using
    \thmref{high-dim-approx-net}, in place of $r$-nets.  As in
    \lemref{greedy_spread}, the approximate $r$-net algorithm needs to
    be modified to mark near-neighbors of previously selected points
    as used, so that they are not selected as part of the net; this
    step does not increase the total running time for the approximate
    $r$-net construction.  The rest of the analysis remains the same.
\end{proof}

\subsection{Approximating all-pairs min-max paths}
\seclab{approx-mst}

As a tool for eliminating the dependence on the spread in our
approximate greedy permutation algorithm, we will use an approximation
to the minimum spanning tree. However, we do not wish to approximate
the total edge length of the tree, as has been claimed by Andoni and
Indyk~\cite{ai-nohaa-06}; rather, we wish to approximate a different
property of the minimum spanning tree, the fact that for every two
vertices it provides a path that minimizes the maximum edge length
among all paths connecting the same two vertices.

\InFullVer{%
   \begin{lemma}[\cite{ai-nohaa-06}]%
       \lemlab{lem:approx-nn}%
       Given $n$ points in a Euclidean space of dimension~$d$, and
       given a parameter $c>1$, we may preprocess the points in time
       and space $O\pth{dn+n^{1+1/c^2+o(1)}}$ into a data structure
       that may be used to answer $c$-approximate nearest neighbor
       queries in query time $O\pth{dn^{1/c^2+o(1)}}$ with high
       probability of correctness.
   \end{lemma}%
}

\begin{lemma}
    \lemlab{approx-mst}%
    Given $n$ points in a Euclidean space of dimension~$d$, and given
    a parameter $\eps>0$, we may in expected time
    $O\pth{n^{1+1/c^2+o(1)}}$ find a spanning tree $T$ of the points
    such that, for every two points $u$ and $v$, the maximum edge
    length of the path in $T$ from $u$ to $v$ is at most $(1+\eps)$
    times the maximum edge length of the path in the minimum spanning
    tree from $u$ to $v$.
\end{lemma}

\InFullVer{%
   \begin{proof}
       As before, we may assume without loss of generality that
       $d=O\pth{\log n/\eps^2}$. We build $T$ by an approximate
       version of \Boruvka's algorithm, in which we maintain a forest
       (initially having a separate one-node tree per point) and then
       in a sequence of $O(\log n)$ stages add to the forest the edge
       connecting each tree with its (approximate) nearest neighbor.
    
       In each stage, we assign each tree of the forest an
       $O(\log n)$-bit identifier. For each pair $(i,b)$ where $i$ is
       one of the $O(\log n)$ bit positions of these identifiers and
       $b$ is zero or one, we build a $(1+\eps)$-approximate nearest
       neighbor data structure for the points whose tree has $b$ in
       the $i$\th bit of its identifier, for a total of $O(\log n)$
       structures. Then, for each point $p$ of the input set, we use
       these data structures to find $O(\log n)$ candidate neighbors
       of $p$, one for each of the $O(\log n)$ structures that do not
       contain $p$. This gives us a set of $O(n\log n)$ candidate
       edges, among which we select for each tree of the current
       forest the shortest edge that has one endpoint in that tree. As
       in \Boruvka's algorithm, with an appropriate tie-breaking rule,
       adding the selected edges to the forest does not produce any
       cycles, and reduces the number of trees in the forest by at
       least a factor of two, so after $O(\log n)$ stages we will have
       a single tree, which we return as our result.
    
       To show that this tree has the desired approximation property,
       consider a complete graph $K_n$ on the input points, in which
       the weight of an edge that does not belong to the output tree
       is the distance between the points. However, in this graph, we
       set the weight of an edge $e$ that does belong to the tree by
       letting $T_1$ and $T_2$ be the two trees containing the
       endpoints of $e$ in the last stage of the algorithm for which
       these endpoints belonged to different trees, and setting the
       weight of $e$ to be the minimum distance between a pair of
       vertices one of which belongs to $T_1$ or $T_2$ and the other
       of which does not belong to the same tree. Then, by the
       correctness of the approximate nearest neighbor data structure,
       the weight of $e$ is at most equal to the distance between its
       endpoints and at least equal to that distance divided by
       $1+\eps$. However (with an appropriate tie-breaking rule) the
       algorithm we followed to construct our tree $T$ is exactly the
       usual version of \Boruvka's algorithm as applied to the
       weighted graph $K_n$.  Therefore, for every $u$ and $v$, the
       path from $u$ to $v$ in $T$ exactly minimizes the maximum edge
       length among all paths in $K_n$, and thus has the desired
       approximation for the original distances.
   \end{proof}
}%
\InShortVer{%
   The proof (in the full version of this paper) modifies \Boruvka's
   algorithm for minimum spanning trees by using an approximate
   nearest neighbor data structure to find candidate edges out of each
   subtree in each phase of the algorithm.%
}%

\subsection{Eliminating the dependence on the spread}

As in \secref{elim-spread}, we will eliminate the $\log\Spread$ term
in the running time for our approximate greedy permutation algorithm
by, in effect, contracting and uncontracting edges of a graph, the
approximate minimum spanning tree of \secref{approx-mst}.

\begin{theorem}
    Let $d\geq 1$, let $X$ be a set of $n$ points in $\Re^d$, and let
    $\eps > 0$, Then, one can compute in
    $O\pth{\eps^{-2} n^{1+1/(1+\eps)^2 + o(1)}}$ expected time a
    sequence that is a $(1+\eps)$-greedy permutation for the Euclidean
    metric on $X$, with high probability.
\end{theorem}

\begin{proof}
    We maintain a partition of the input into subproblems, defined by
    subtrees of the spanning tree $T$ computed by
    \lemref{approx-mst}. Initially, there is one subproblem, defined
    by the whole tree, but it does not include all the input
    points. Rather, as the algorithm progresses, certain points within
    each subproblem's subtree will be active, depending on the current
    value $r$ for which we are computing approximate $r$-nets.
    
    We delete edge $e$ from tree $T$, splitting its subproblem into
    two smaller sub-subproblems, whenever $r$ becomes smaller than the
    length of $e$ divided by $1+O(\eps)$. %
    \InFullVer{%
       After this point, the points on one side of $e$ are too far
       away to affect the choices made on the other side of $e$. %
    }%
    We will include the endpoints of an edge $e$ into the active
    points of the subproblem containing $e$ whenever the current value
    of $r$ becomes smaller than $cn/\eps$ times the length of $e$ (for
    an appropriate constant $c<1$).  \InFullVer{%
       Until that time there will always be another active point
       within distance $c\eps r$ of $e$, so omitting the endpoints of
       $e$ will not significantly affect the approximation quality of
       the greedy permutation we construct.%
    }%
    
    At each stage of the algorithm, when we construct approximate
    $r$-nets for some particular value of $r$, we do so separately in
    each subproblem that has two or more active points. (In a
    subproblem with only one active point, that point must have been
    included in the approximate greedy permutation already, and so
    cannot be chosen for the new $r$-net. However, the points that
    have been included in the permutation must remain active in order
    to prevent other points within distance $r$ of them from being
    added to the net.)  Each edge $e$ contributes to the size of a
    subproblem only for a logarithmic number of different values of
    $r$, so the total time is as stated.
\end{proof}

\section{Exact greedy permutation for bounded %
   treewidth}
\seclab{e:g:treewidth}

\InFullVer{%
   For restricted classes of graphs, we can compute a greedy
   permutation exactly, more quickly than the naive algorithm of
   \secref{naive}.  Our algorithms follow Cabello and
   Knauer~\cite{ck-agbtors-09} in applying orthogonal range searching
   data structures to the vectors of distances from small sets of
   separator vertices. Specifically, we need the following result.
   
   \begin{lemma}[\cite{gbt-srtgp-84}]%
       \lemlab{post-office}%
       Let $S$ be a set of $n$ points in $\Re^k$, for a constant
       $k>1$.  Then we may process $S$ in time and space
       $O\pth{n\log^{k-1}n}$ so that the $\ell_\infty$-nearest
       neighbor in $S$ of any query point~$q$ may be found in query
       time $O\pth{\log^{k-1} n}$.
   \end{lemma}
   
   \subsection{Tree decomposition and restricted partitions}
   
   A \emphi{tree decomposition} of a graph $\Graph$ is a tree
   $\mathcal{D}$ whose nodes are associated with sets of vertices of
   $\Graph$ called \emph{bags}, satisfying the following two
   properties:
   \begin{compactenum}[\qquad(A)]
       \item For each vertex $v$ of $\Graph$, the bags containing $v$
       induce a connected subtree of $\mathcal{D}$.
       \item For each edge $uv$ of $\Graph$, there exists a bag in
       $\mathcal{D}$ containing both $u$ and $v$.
   \end{compactenum}
   The \emphi{width} of a tree decomposition is one less than the
   maximum size of a bag, and the \emphi{treewidth} of a graph is the
   minimum width of any of its tree decompositions.
   
   Following Frederickson~\cite{f-adsdecksst-97}, who defined a
   similar concept on trees, we define a \emphi{restricted order-$k$
      partition} of a graph $\Graph$ of treewidth $w$ to be a
   partition of the edges of $\Graph$ into $O(n/k)$ subgraphs $S_i$
   such that, for all $i$, the subgraph $S_i$ has at most $k$ edges,
   and such that, for each $S_i$, at most $2w+2$ vertices are
   endpoints of edges both in and outside $S_i$.
   
   \begin{lemma}
       \lemlab{lem:respart}%
       Fix $w=O(1)$.  Then for every $n$-vertex graph $\Graph$ of
       treewidth $\le w$ and every $k\ge\tbinom{w}{2}$, a restricted
       order-$k$ partition of $\Graph$ may be constructed in time
       $O(n)$ from a tree decomposition of~$\Graph$.
   \end{lemma}
   
   \begin{proof}
       Let $\mathcal{D}$ be a tree decomposition of $\mathcal{D}$, of
       width~$w$.  Without loss of generality (by choosing a root
       arbitrarily and splitting some nodes of $\mathcal{D}$ into
       multiple nodes having equal bags) we may assume that
       $\mathcal{D}$ is a rooted binary tree.  Associate each edge of
       $\Graph$ with a unique node of $\mathcal{D}$ having the two
       endpoints of the edge in its bag, choosing arbitrarily if there
       is more than one node.
       
       Following Frederickson~\cite{f-adsdecksst-97}, find a partition
       of the nodes of $\mathcal{D}$ into subsets $D_i$, having the
       following properties:
       \begin{compactenum}[\qquad(A)]
           \item Each subset $D_i$ induces a connected subtree of
           $\mathcal{D}$.
           \item Each subset $D_i$ is associated with at most $k$
           edges of $\Graph$.
           \item For each $i$, if $D_i$ contains more than one node,
           then at most two edges of $\mathcal{D}$ have one endpoint
           in $D_i$ and the other endpoint in a different subset.
           \item No two adjacent subsets $D_i$ can be merged while
           preserving properties (A), (B), and (C).
       \end{compactenum}
       This partition may be found in linear time by a greedy
       bottom-up traversal of $\mathcal{D}$.  For each subset $D_i$,
       we form a subgraph $S_i$ of $\Graph$ consisting of the edges
       associated with nodes in $D_i$. Property (B) implies that each
       subgraph $S_i$ has at most $k$ edges, and property (C) implies
       that there are at most $2w+2$ vertices that are endpoints of
       edges both in $S_i$ and in other subgraphs (namely the vertices
       in the bags of the two nodes of $D_i$ that are endpoints of
       edges connecting $D_i$ to other sets).
       
       It remains to show that there are $O(n/k)$ subsets.
       Contracting each subset $D_i$ in $\mathcal{D}$ to a single node
       produces a binary tree $\Tree$ in which each edge with one
       degree-one endpoint or two degree-two endpoints connects
       subgraphs that together have more than $k$ edges of $\Graph$,
       for otherwise these two subgraphs would have been merged. It
       follows that $\Tree$ has $O(n/k)$ such edges, and therefore
       that it has only $O(n/k)$ nodes. Thus, we have formed $O(n/k)$
       subsets $S_i$, as desired.
   \end{proof}
   
   Given a restricted partition of $\Graph$, we define an
   \emph{interior vertex} of a subgraph $S_i$ to be a vertex of
   $\Graph$ all of whose incident edges belong to $S_i$, and we define
   a \emph{boundary vertex} of $S_i$ to be a vertex that is incident
   to an edge in $S_i$ but is not an interior vertex of $S_i$.
   
   \subsection{The algorithm}
   Suppose graph $\Graph$ has treewidth $w=O(1)$. We may find a greedy
   permutation of $\Graph$ by performing the following steps.
   \begin{compactenum}[\quad 1.]
       \item Find a tree-decomposition of $\Graph$ of width
       $w$~\cite{b-ltaftdst-96}.
       \item Apply \lemref{lem:respart} to find a restricted
       order-$\sqrt n$ partition of $\Graph$.
       \item Construct a weighted graph $H$ whose vertices are the
       boundary vertices of the restricted partition, where two
       vertices $u$ and $v$ are connected by an edge if they belong to
       the same subgraph $S_i$ of the partition, and where the weight
       of this edge is the length of the shortest path in $S_i$ from
       $u$ to $v$ (or a sufficiently large dummy value $Z$, greater
       than $n$ times the heaviest edge of $\Graph$, if no such path
       exists).
       \item For each vertex $v$ of $H$, initialize a value $d(v)$ to
       $Z$. As the algorithm progresses, $d(v)$ will represent the
       length of the shortest path to $v$ from a vertex already
       belonging to the greedy permutation.
       \item For each subgraph $S_i$ of the restricted partition,
       having $k$ boundary vertices, construct a $(k+1)$-dimensional
       $\ell_\infty$-nearest neighbor data structure
       (\lemref{post-office}) whose points correspond to interior
       vertices of $S_i$. The first coordinate of the point for $v$ is
       the length of the shortest path within $S_i$ to $v$ from an
       interior vertex of $S_i$ that belongs to the greedy
       permutation; initially, and until such a path exists, it is
       $Z$. The remaining coordinates give the lengths of the shortest
       paths in $S_i$ from each boundary vertex to $v$, or $Z$ if no
       such path exists.
       \item Repeat $n$ times:
       \begin{compactenum}[(a)]
           \item For each subgraph $S_i$ of the restricted partition,
           find the farthest interior vertex of $S_i$ from the
           already-selected vertices, and its distance from the
           vertices selected so far. This may be done by querying the
           nearest neighbor of a point $q$ whose first coordinate is
           $2Z$ and whose $i$\th coordinate is $2Z-d(v_i)$, where
           $v_i$ is the $i$\th boundary vertex of $S_i$.
           \item Among the $O\pth{\sqrt n}$ vertices found in the
           previous step, and the $O\pth{\sqrt n}$ boundary vertices
           whose distance $d(v)$ from the greedy permutation was
           already known, select a vertex $v$ whose distance from the
           greedy permutation is maximum, and add it to the
           permutation.
           \item If $v$ is an interior vertex of a subgraph $S_i$ of
           the restricted partition, use Dijkstra's algorithm to
           compute the shortest path within $S_i$ from it to the other
           interior vertices, and rebuild the nearest neighbor data
           structure associated with $S_i$ after updating the first
           coordinates of each of its points.
           \item Use Dijkstra's algorithm on $H$ (if $v$ is a boundary
           vertex) or on $H+S_i$ (if it is interior to~$S_i$),
           starting from $v$, to update the values $d(u)$ for each
           boundary vertex $u$.
       \end{compactenum}
   \end{compactenum}
   
   \smallskip\noindent The time analysis of this algorithm is
   straightforward and gives us the following result.
}
   
   \begin{theorem}
       Let $\Graph$ be an $n$-vertex graph of treewidth $w=O(1)$ with
       non-negative edge weights. Then in time
       $O\pth{n^{3/2}\log^{2w+2} n}$ we may construct an exact greedy
       permutation for $\Graph$.
   \end{theorem}%

\section{Counting distances in planar graphs}
\seclab{planar:counting}

In this section we give a near-linear time bicriterion approximation
algorithm for counting pairs of vertices in a planar graph with a
given pairwise distance $r>0$.  The result is approximate in the
following sense.  If we let $c$ and $c'$ be the number of pairs with
distance at most $r$ and at most $(3+\eps)r$ respectively, for some
$\eps>0$, then we output a number $\alpha\in [c,c']$.

The following is due to Thorup \cite{t-corad-04} (see also
\cite{kst-mcoad-13}).

\begin{theorem}[Thorup \cite{t-corad-04}]%
    \thmlab{oracle}%
    For any $n$-vertex undirected planar graph with non-negative edge
    lengths, and for any $\eps>0$, there exists a
    $(1+\eps)$-approximate distance oracle with query time
    $O(\eps^{-1})$, space $O(\eps^{-1} n \log n)$, and preprocessing
    time $O(\eps^{-2} n \log^3 n)$.
\end{theorem}

\newcommand{\decomp}{{\cal C}}%
\newcommand{\HD}{{\cal H}}%
\newcommand{\GInduced}[1]{\Graph[#1]}

The basic idea is now to recursively decompose $\Graph$ using planar
separators. Fortunately, one can do it in such a way, that when
looking on a patch $P$, with $m$ vertices, formed by this recursive
decomposition, the distances between the boundary vertices of $P$ (in
the original graph) are known. The details of how to compute this
decomposition is described by Fakcharoenphol and Rao
\cite{fr-pgnwe-06}, and we recall their result.

\newcommand{\VNode}[1]{C\pth{#1}}

Let $\Graph=(\Vertices, \Edges)$ be a graph.  A \emphi{patch} of a
graph is a subset $C \subseteq \Vertices$, such that the induced
subgraph $\GInduced{C}$ is connected. A vertex $v \in C$ is a
\emphi{boundary} vertex if there exists a vertex
$u \in \Vertices \setminus C$ such that $uv \in \Edges$.  A
\emphi{hierarchical decomposition} $\HD$ of $\Graph$ is a set of
subsets of the vertices of $\Graph$, that can be described via a
binary tree $\Tree$, having patches of $\Graph$ associated with each
node of it. The root $r$ of tree is associated with the whole graph;
that is, $\VNode{r} = \Vertices$. Every node of $u \in \Tree$ has two
children $v_1, v_2$, such that
$\VNode{v_1} \cup \VNode{v_2} = \VNode{u}$, and
$\cardin{\VNode{v_1}}, \cardin{\VNode{v_2}} \leq (2/3)\cardin{
   \VNode{u}}$.
A leaf of this tree is associated with a single vertex of
$\Graph$. Finally, we require that for any patch $C$ in this
decomposition, the set of its boundary vertices $\bdX{C}$, has at most
$O\pth{\bigl.\smash{\sqrt{\cardin{C}}}\, }$ vertices.

For every $C\in \HD$, the (inner) \emphi{distance graph} of $C$,
denoted by $\DGraph{C}$ is a clique over $\bdX{C}$, with
$u,v \in \bdX{C}$ assigned length $d_{\Graph[C]}(u,v)$, i.e.~equal to
the shortest distance of all paths between $u$ and $v$, that are
contained entirely inside $\GInduced{C}$. %
\remove{Similarly, the \emphi{outer distance graph} of $C$, denoted by
   $\DOutGraph{C}$ is the graph over the boundary vertices of $C$,
   with the distance between two vertices $u,v \in \bdX{C}$ being the
   shortest distance between $u$ and $v$ in $\Graph$ that does not use
   an vertex of $C \setminus \bdX{C}$. }%
The \emph{dense distance graph} associated with $\HD$ is the graph
$\Graph' = \bigcup_{C \in \decomp} \DGraph{C}$.

\begin{theorem}[Fakcharoenphol and Rao \cite{fr-pgnwe-06}]%
    \thmlab{FT}%
    Let $\Graph$ be an $n$-vertex undirected planar graph with
    non-negative edge lengths.  Then, one can compute, in
    $O(n \log^3 n)$ time, a hierarchical decomposition $\HD$ of
    $\Graph$, and all the inner and outer distance graphs associated
    with its patches (i.e., $\DGraph{C}$ for all $C \in \HD$).
\end{theorem}

We are now ready to prove the main result of this section.

\begin{theorem}
    Let $\Graph$ be a given $n$-vertex undirected planar graph with
    non-negative edge lengths, let $r>0$, and $\eps>0$.  Let
    $c = \cardin{\ShortPairs{r}}$ and
    $c'= \cardin{ \ShortPairs{(3+\eps)r} }$, see
    \Eqrefpage{short:pairs}.  Then, one can compute, in
    $O(\eps^{-2} n \log^3 n)$ time, an integer $\alpha$, such that
    $c\leq \alpha \leq c'$.
\end{theorem}
\begin{proof}
    We compute a hierarchical decomposition of $\HD$, and for every
    $C \in\HD$, the distance graph $\DGraph{C}$, in total time
    $O(n\log^3 n)$, using \thmref{FT}.  We consider all patches
    $C \in \HD$. If $\cardin{C}=1$ then we let $ \alpha_\decomp = 0$.
    Otherwise, our purpose here is to count the number of pairs of
    vertices $u,v \in C$, such that $\distG{u}{v} \leq r$, and $u$ and
    $v$ belong to different children of $C$.

    So, let $C_1,C_2$ be the two patches that are the children of $C$
    in $\HD$.  Let $B_1$, $B_2$ be the set of border vertices of
    $C_1$, $C_2$, respectively, and let $B=B_1\cup B_2$.  Let
    $G_1 = \Graph[C_1] \cup \DGraph{C_2}$, and
    $G_2 = \Graph[C_2] \cup \DGraph{C_1}$.  That
    is, $G_1$ is the union of the subgraph of $\Graph$ induced on the
    cluster $C_1$, and the distance graph inside $C_2$, and the
    distance graph outside $C_1 \cup C_2$.  Note that
    \begin{inparaenum}[(i)]
        \item $V(\Graph_1) = C_1$, 
        \item for any $u,v\in V(G_1)$, we have
        $d_{G_1}(u,v)=d_\Graph(u,v)$, and
        \item
        $\cardin{\EdgesX{\Graph_1}} = O({\cardin{C_1}}) +
        O\pth{\cardin{\bdX{C_2}}^2} + O\pth{\cardin{\bdX{C}}^2} =
        O({\cardin{C_1}}) $.
    \end{inparaenum}
    Therefore, for any $i\in\{1,2\}$, by running Dijkstra's algorithm
    on $G_i$ starting from $B$, we can compute the set of vertices
    \InFullVer{%
    \begin{align*}
        U_i = \Set{v\in C_i}{ d_G(v,B) \leq r},
    \end{align*}
    }
    \InShortVer{%
    \begin{math}
        U_i = \Set{v\in C_i}{ d_G(v,B) \leq r},
    \end{math}
    }
    in time $O(|C_i| \log n)$.  Moreover, for any $i\in\{1,2\}$, for
    any $v\in C_i$, we can compute (within the same time bounds) its
    closest vertex in $B$; specifically, a 
    vertex $\Gamma_i(v)\in B$, with $d_G(v,\Gamma_i(v)) = d_G(v,B)$.
    Let $B'$ be the set of all vertices that are border vertices in
    some ancestor of $\decomp$ in $\HD$.  For any $u\in B$,
    $i\in\{1,2\}$, let
    \InFullVer{%
       \begin{align*}
           U_i(u) = \Set{v\in C_i}{\bigl. \Gamma_i(v) = u } \setminus B'.
       \end{align*}%
    }%
    \InShortVer{%
       \begin{math}
           U_i(u) = \Set{v\in C_i}{\bigl. \Gamma_i(v) = u } \setminus B'.
       \end{math}
    }%
    Using \thmref{oracle} we can find in time $O(\eps^{-1}|B|^2)$ all
    pairs of vertices $u,v\in B$, with $d_G(u,v)\leq (1+\eps)r$.  That
    is, we compute the set of border vertex pairs
    \begin{align*}
        T = \Set{(u,v)\in B \times B }{\bigl. d_G(u,v)\leq (1+\eps)r
        }.
    \end{align*}
    We now explicitly count all the pairs of vertices that are in
    distance at most $r$ from a pair of vertices on the boundary, such
    that these boundary vertices in turn are in distance at most
    $(1+\eps)r$ from each other. That is, we set
    \begin{align*}
        \alpha^{}_{\decomp} = \sum_{\{u,v\}\in
           T}\pth{\Bigl. |U_1(u)|\cdot |U_2(v)| + |U_2(u)|\cdot
           |U_1(v)|}.
    \end{align*}
    Finally, we return the value
    $ \alpha = \sum_{\decomp \in \HD} \alpha_\decomp.  $ Every
    ordered pair $(u,v)\in V(G)\times V(G)$ is counted exactly once in
    the above summation.  Moreover, every pair with $d_G(u,v)\leq r$
    contributes 1, while every pair with $d_G(u,v)>(1+\eps) 3 r$
    contributes 0, implying that $c\leq \alpha \leq c'$, as required.
    
    It remains to bound the running time.  Constructing $\HD$ takes
    time $O(n \log^3 n)$.  The hierarchical decomposition has
    $O(\log n)$ levels.  For every level, we spend a total of
    $O(n\log n)$ time running Dijkstra's algorithm.  We also spend a
    total time of $O(\eps^{-1} n)$ computing the sets $T$.  Finally,
    we spend $O(\eps^{-2} n \log^3 n)$ time at the preprocessing step
    of the distance oracle from \thmref{oracle}.  Therefore, the total
    running time is $O(\eps^{-2} n \log^3 n)$, which concludes the
    proof.
\end{proof}

\section{Conclusions}

We have found efficient approximation algorithms for greedy
permutations in graphs and in high-dimensional Euclidean spaces, and
fast exact algorithms for graphs of bounded treewidth.
This implies $(2+\eps)$-approximate $k$-center clustering of graph
metrics in $O_\eps( m \log^2 n)$ time (ignoring the dependency on
$\eps$), for all values of $k$ simultaneously. For a single value of
$k$, and for graphs whose weights are positive integers, our technique
can be used to construct a $2$-approximation to the $k$-center
clustering in $O\pth{ m \log n \log \SpreadX{ G} }$ expected time
\InFullVer{(\thmref{k:center:intger:spread}).}%
\InShortVer{(Section~2.4 in the full version).} %
This compares favorably with a significantly more complicated
algorithm of Thorup \cite{t-qmcflsg-05} that has running time worse by
at least a factor of $O( \log^3 n)$.

We leave open for future research finding other graphs in which we may
construct exact greedy permutations more quickly than the naive
algorithm. Another direction for research concerns hyperbolic spaces
of bounded dimension. Krauthgamer and Lee~\cite{kl-ancs-06} claim
without details that for all $\eps>0$, sets of $n$ points in
hyperbolic spaces of bounded dimension have $(1+\eps)$-Steiner
spanners with $O(n)$ vertices and edges.  Applying our graph algorithm
to these spanners (modified to avoid including Steiner points in its
$r$-nets) would give a near-linear time approximate greedy permutation
for those spaces as well, but perhaps a more direct and more efficient
algorithm is possible.




\InShortVer{\newpage}

\newcommand{\etalchar}[1]{$^{#1}$}
\providecommand{\CNFX}[1]{
  {\em{\textrm{(#1)}}}}\providecommand{\CNFSODA}{\CNFX{SODA}}\providecommand{\CNFFOCS}{\CNFX{FOCS}}\providecommand{\CNFCCCG}{\CNFX{CCCG}}\providecommand{\CNFSoCG}{\CNFX{SoCG}}\providecommand{\CNFSTOC}{\CNFX{STOC}}
  \providecommand{\CNFPODS}{\CNFX{PODS}}


\end{document}